\documentclass{article}
\usepackage[utf8]{inputenc}
\usepackage[T1]{fontenc}
\usepackage[english]{babel}

\usepackage[a4paper, left=20mm, right=20mm, top=2cm, bottom=2cm]{geometry}

\usepackage{amsmath,amssymb,amsthm,mismath}
\usepackage{thm-restate}
\usepackage{algorithm,algcompatible}

\newcommand{\bra}[1]{\mathinner{\left \langle #1 \right\rvert}}
\newcommand{\ket}[1]{\mathinner{\left \lvert #1 \right\rangle}}
\DeclarePairedDelimiterX\ketbra[2]{| }{|}{#1 \delimsize\rangle\!\delimsize\langle #2}
\providecommand{\nrm}[1]{\left\lVert #1 \right\rVert}

\newtheorem{lemma}{Lemma}
\newtheorem{corollary}{Corollary}

\usepackage[affil-it]{authblk}
\usepackage[dvipsnames,svgnames,x11names]{xcolor}
\usepackage[pdftex,bookmarks,pagebackref=true, draft=false]{hyperref}
\hypersetup{ % a color scheme often used in our drafts - but feel free to change as desired
    bookmarksnumbered=true, % If Acrobat bookmarks are requested, include section numbers
    unicode=false, 
    pdfstartview={FitH}, % fits the width of the page to the window
    pdftitle={Lower bounds on hybrid quantum classical query algorithms}, % title
    pdfauthor={Andris Ambainis, András Gilyén, Martins Kokainis}, % author
    pdfsubject={}, % subject of the document
    pdfcreator={}, % creator of the document
    pdfproducer={}, % producer of the document
    pdfkeywords={}, % list of keywords
    pdfnewwindow=true, % links in new window
    colorlinks=true, % false: boxed links; true: colored links
    urlcolor=WildStrawberry, % color of external links
    linkcolor=blue, % color of internal links
    citecolor=ForestGreen, % color of links to bibliography
    filecolor=BrickRed % color of file links
}

\usepackage{mleftright}\mleftright % fixes bracket spacing
\usepackage{cleveref}
\newcommand{\raw}{\mathrm{raw}}
\newcommand{\inlabel}{\mathrm{in}}
\newcommand{\outlabel}{\mathrm{out}}

\makeatletter
\newcommand{\shiftdown}[2][0.5ex]{\mathpalette{\shiftdown@{#1}}{#2}}
\newcommand{\shiftdown@}[3]{\raisebox{-#1}{$\m@th#2#3$}}
\makeatother

\usepackage{ifdraft}
\usepackage{comment}
\ifdraft{
	\newcommand{\authnote}[3]{{\color{#3} {\bf  #1:} #2}}
}{
	\newcommand{\authnote}[3]{}
}

\begin{document}
\title{Tight bounds for hybrid quantum-classical query algorithms}
\author[1]{Andris Ambainis}
\author[2]{András Gilyén}
\author[1]{Martins Kokainis}
\affil[1]{Center for Quantum Computer Science, Faculty of Computing,
University of Latvia} 
\affil[2]{Alfréd Rényi Institute of Mathematics, Budapest, Hungary} 
\date{}
\maketitle

\begin{abstract}
We study hybrid quantum-classical algorithms in the query model. The computation consists of quantum subroutines that make at most $q$ oracle queries; between subroutines, all qubits are measured and discarded. We prove matching upper and lower bounds for several problems in this model:
\begin{itemize}
    \item (Unbiased) Phase estimation up to a precision $\epsilon$ requires $\Theta(\frac{1}{q \epsilon^2})$ queries;
    \item (Unbiased) Amplitude estimation up to a precision $\epsilon$ requires $\Theta(\frac{p(1-p)}{q \epsilon^2})$ queries;
    \item Search among $N$ items requires $\Theta(\frac{N}{q})$ queries;
    \item Two level AND-OR tree (AND of $m$ ORs, with $n$ inputs to each OR) requires $\Theta(\frac{nm}{q})$ queries.
\end{itemize}
All of these bounds are optimal up to a constant factor, for all $q$ from 1, corresponding to the classical complexity, to $Q(f)$, the unrestricted quantum query complexity of the respective problem. We also derive a hybrid lower bound for distinguishing two distributions specified by a state-preparation unitary; this bound is tight up to a logarithmic factor.

Besides specific lower bounds, an important contribution is developing methods for proving lower bounds on hybrid quantum algorithms (which have been very ad-hoc up to now). The first two bounds follow from a common framework: we analyze probability distributions over measurement transcripts and bound a progress measure that captures their distinguishability. The AND-OR bound requires a more delicate argument, combining two progress measures that track the information available to the algorithm classically and in quantum form, respectively.
\end{abstract}

\section{Introduction}

Hybrid quantum-classical computation is a major paradigm for near-term quantum computing. Limitations on coherence time and error correction may prevent the execution of a single long quantum computation, but one can try to compensate for that
by running a sequence of short quantum computations, with later computations chosen on the basis of earlier measurement outcomes. This makes it important to understand the power of algorithms built from many short quantum subroutines.

The computational power of this model is not immediately clear. Classical feedback can make the sequence of quantum computations highly adaptive, but measuring between subroutines destroys the quantum coherence that is responsible for the quantum advantage. A natural question is how the power of a hybrid algorithm depends on the number of coherent quantum steps that can be performed in a subroutine. In this paper, we study this question in the query model, which is commonly used for analyzing quantum algorithms. We use $q$ to denote the maximum number of queries performed within one coherent quantum block; varying $q$ therefore quantifies the tradeoff between the length of a coherent computation and total query complexity.

There are two complementary ways to study this question. One is to seek general characterizations of the model, asking, for example, how large a quantum advantage is possible when every coherent block contains at most $q$ queries. The other is to study specific computational tasks, such as search, counting, or element distinctness, and determine how their complexity varies with $q$. The first perspective describes the ultimate power and limitations of the model, while the second reveals which of the familiar quantum speedups remain robust under restricted coherence.

From the first perspective, Simon's problem \cite{simon1997SimonsAlgorithm} or Forrelation \cite{AaronsonAmbainis2018} provide examples where exponential or superexponential advantage can be achieved even with 1-query blocks. 
For total functions, a polynomial advantage is possible: Sun and Zheng \cite{sun2019-hybrid} constructed a total Boolean function $f$ satisfying
\[
Q(f;1)=\widetilde{\bigO}\bigl(R(f)^{4/5}\bigr),
\]
where $R(f)$ denotes the randomized query complexity and $Q(f;1)$ denotes the
quantum query complexity with blocks of size 1. The dependence on $q$ can be even more dramatic. The oracle separation of Hasegawa and Le Gall \cite{HasegawaLeGall2022}, although formulated in terms of quantum depth, can be recast in hybrid query language as giving problems with polynomial query complexity for $q=d+1$ but exponential query complexity for $q=d$.

From the second perspective, possible strategies for splitting up 
a quantum algorithm into shorter quantum subroutines have been explored for various problems, from basic search \cite{sun2019-hybrid} to amplitude estimation \cite{giurgicaTiron2020LowDepthAE} to AND-OR trees \cite{murcca2024two}. The lower-bound side, however, is much less developed than in the standard quantum query model, where powerful general techniques such as the adversary method \cite{andor-ambainis,reichardt2011reflections} and the polynomial method \cite{beals2001QLowerBoundPoly} are available.

In the hybrid case, only a few nontrivial quantum lower bounds are known.
Sun and Zheng \cite{sun2019-hybrid} have shown a tight quantum lower bound for search but the argument is heavily tailored to Grover's search problem. 
A more broadly applicable method is the hybrid version of the compressed oracle method \cite{hamoudi2024nisq} which gives tight lower bounds for search and collision-finding. However, similarly to the original compressed-oracle method \cite{Zhandry19}, it is mostly useful when the hard input distribution is a product distribution (that is, when the input variables are independent).

Our goal is to develop a systematic approach for proving hybrid query lower bounds (similarly to the role of the adversary method in the standard quantum query model). We develop a framework which gives three main results. First, we obtain a general lower bound for binary classification of unitaries which gives optimal tradeoffs for phase and amplitude estimation and near-optimal bounds for distribution distinguishing with purified access. Second, the same framework gives a new proof of the tight lower bound for search. Third, by combining this framework with ideas from the compressed-oracle method, we prove the first optimal lower bound for the two-level AND-OR tree. Thus, our results include both new lower bounds and unified proofs of previously known ones.

The first main result is the lower bound for binary classification. Let ${\cal Q}_0, {\cal Q}_1$ be two sets of unitaries acting on the same finite dimensional space. Given query access to an unknown unitary $U \in {\cal Q}_0 \cup {\cal Q}_1$, the task is to  determine whether $U$ is in ${\cal Q}_0$ or ${\cal Q}_1$.
\begin{restatable}[Hybrid lower bound for binary classification]{theorem}{classification}\label{thm:classification}
Assume that quantum subroutines use at most a total of $q$ (controlled) queries to $U$ and $U^\dagger$. If there exist $U_0 \in {\cal Q}_0$ and $U_1 \in {\cal Q}_1$ for which $\nrm{U_0 - U_1} = \epsilon$, then any hybrid algorithm that solves this binary classification problem with bounded error requires $\Omega(\frac{1}{\epsilon}+\frac{1}{q \epsilon^2})$ queries in the worst case.
\end{restatable}

This includes phase estimation and amplitude estimation as special cases.

\begin{corollary}[Hybrid query complexity of phase estimation]\label{cor:phaseEstLower}
Suppose we are given (controlled) query access to an unknown unitary $U$, and we are also given an eigenstate $U \ket{\psi} = e^{i \phi} \ket{\psi}$. In our hybrid model, the query complexity of estimating $\phi$ to precision $\epsilon \in (0, 1)$ with success probability $\ge \frac 2 3$ is $\Theta(\frac{1}{\epsilon}+\frac{1}{q \epsilon^2})$. 
\end{corollary}
\begin{proof}
The lower bound follows from \Cref{thm:classification} by setting $Q_0 = \{I\}$ and $Q_1 = \{I + (e^{i \epsilon} - 1) \ketbra{0}{0}\}$, while the upper bound can be achieved, e.g., via repeated unbiased phase estimation with precision $\bigO(\frac 1 q)$, see~\Cref{cor:unbiasedPE}.
\end{proof}

\begin{corollary}[Hybrid query complexity of amplitude estimation]\label{cor:ampEstLower}
Suppose we are given (controlled) query access to an unknown unitary $U$ such that $U \ket{\bar{0}} = \sqrt{s}\ket{1}\ket{\psi_1}+\sqrt{1-s}\ket{0}\ket{\psi_0}$, and the only thing we are promised is that $s(1-s)\leq p(1-p)$ for some known $p\in(0,1)$. In our hybrid model the query complexity of estimating $s$ to precision $\epsilon \leq p(1-p)$ with success probability $\ge \frac 2 3$ is $\Theta\Big(\frac{\sqrt{p(1-p)}}{\epsilon}+\frac{p(1-p)}{q \epsilon^2}\Big)$. 
\end{corollary}
\begin{proof}
Assume without loss of generality that $p\leq 1/2$. The lower bound follows from \Cref{thm:classification} by setting $Q_i = \{\exp(-i\phi_i Y )\}$ for the Pauli $Y$ matrix and $\phi_i=\arcsin(\sqrt{p-i\epsilon})$ so that $U_i\ket{0}=\sqrt{p-i\epsilon}\ket{1}+\sqrt{1-(p-i\epsilon)}\ket{0}$ and $\nrm{U_0-U_1}=\Theta(\epsilon/\!\sqrt{p})$.
The upper bound can be achieved, e.g, via repeated (unbiased) amplitude estimation with variance $\frac{p(1-p)}{q^2}$,~see~\Cref{cor:unbiasedPE}.
\end{proof}

In both corollaries, the estimator obtaining the upper bound is unbiased (except in the $1\ll q \ll (p(1-p))^{-\frac{1}{2}}$ regime for amplitude estimation), and the algorithm is non-adaptive, so all depth $q$ circuits may be run in parallel. The lower bound applies even to potentially biased estimators executed adaptively.

As a further application, we get an almost tight characterization of the hybrid query complexity of distribution distinguishing with purified access. This access model is fairly generic~\cite{gilyen2019DistPropTestQuant}, and includes situations for example when one obtains distributions via (quantum) Monte Carlo methods~\cite{kothari2022MeanEstimationSourceCode}. The following corollary provides a smooth interpolation between the incoherent sample- and the fully coherent query-complexities~\cite{belovs2019QAlgClassProbDist}. 

\begin{corollary}[Hybrid query complexity of distinguishing distributions given by purified access]\label{cor:Hellinger}
Let $r,s\in\Delta_k$ be discrete probability distributions given via purified access as follows: $U_p$ is an unknown unitary corresponding to a uniformly random $p\in\{r,s\}$ such that $U_p\ket{\bar{0}}=\sum_{i\in[k]}\sqrt{p_i}\ket{i}\ket{\psi_i}$. With controlled access to $U_p$, the hybrid query complexity of distinguishing the cases $U_p=U_r$ and $U_p=U_s$ with success probability at lest $\frac23$ is $\widetilde{\Theta}\Big(\frac1{d_H}+\frac1{q d_H^2}\Big)$, where $d_H$ is the Hellinger distance of the distributions
	\begin{equation}\label{eq:hellinger}
		d_H(r,s) \coloneqq \shiftdown[-0.76ex]{\sqrt{\sum_{i\in[k]}\shiftdown[0.96ex]{\frac{\left(\sqrt{r_i}-\sqrt{s_i}\right)^{\shiftdown[0.2ex]{2}}}{2}}}}.
	\end{equation}
\end{corollary}

\begin{proof}
    The $\Omega(\max(1/(q d_H^2),1/d_H))$ lower bound is tight, already in the $k=2$ case, as it is essentially equivalent to \Cref{cor:ampEstLower}. The upper bound also follows from a reduction to the binary ($k=2$) case via~\cite{pensia2022CommConstrainedHypTesting}, as inspired by~\cite[Appendix~C]{belovs2019QAlgClassProbDist}. Indeed, \cite[Corollary 3.4]{pensia2022CommConstrainedHypTesting} states that there is always a $B\subset[k]$ such that the binary distributions $r^B,\,s^B\in\Delta_2$ defined as $p^B_1=\sum_{i\in B}p_i,\,p^B_2=\sum_{i\in [k]\setminus B}p_i$ satisfy
    \begin{align*}
        1\le\frac{d_H(r,s)}{d_H(r^B,s^B)} \le \bigO\left(\min(k,\log(1/d_H))\right),
    \end{align*}
    therefore we can use the amplitude estimation algorithm of \Cref{cor:ampEstLower} applied to $r^B,\,s^B$, to get an algorithm that is asymptotically tight for constant $k$, and in general has at most a logarithmic slack in $d_H$.
\end{proof}

The lower bound of \Cref{cor:phaseEstLower} was previously shown in the unpublished manuscript \cite{tong2021-phase-depth}.

For search, we consider the standard setting in which we are given $x_1, \ldots, x_N \in \{0, 1\}$ with standard query access to them. 
To prove the lower bound, we restrict ourselves to the case when there is either no $i:x_i = 1$ or exactly one such $i$ and the task is to distinguish between the two cases. 

\begin{restatable}{theorem}{grover}
\label{thm:2}
Assume that quantum subroutines may use at most $q$ queries. Then $\Theta\big(\sqrt{N}+\frac{N}{q}\big)$ queries are necessary and sufficient to distinguish between the cases when there is no $i:x_i = 1$ versus exactly one such~$i$.
\end{restatable}

This bound was first proved by Sun and Zheng \cite{sun2019-hybrid}; here we obtain it through the same framework as the binary-classification bound.

The third problem that we consider is the two-level AND-OR tree.
\[
F(x) = \bigwedge_{i = 1}^{m} \bigvee_{j = 1}^{n} x_{ij},
\qquad x \in \{0, 1\}^{m \times n},
\]
where $m, n \geq 2$ are integers. This models a nested search problem.

We tightly characterize the complexity of $F$ in the hybrid query model.

\begin{restatable}{theorem}{andor}
\label{thm:andor-main}
Assume that quantum subroutines may use at most $q$ queries. 
Then, $\Theta\Big(\sqrt{mn} + \frac{mn}{q}\Big)$ queries are necessary and sufficient to compute $F$.
\end{restatable}

The upper bound in this theorem is due to \cite{murcca2024two}. The lower bound
is shown in this paper for the first time.

{\bf Methods.}
We consider the transcript consisting of the sequence of classical measurement outcomes produced by the successive quantum blocks. For each input, a hybrid algorithm induces a probability distribution on the space of transcripts. If the algorithm distinguishes two inputs with bounded error, the corresponding transcript distributions must have constant statistical distance. We capture this difference by a $\chi^2$-type progress measure.

For binary classification and search, we establish a one-block inequality relating the distinguishability generated within a quantum block (measured by the distance between the final quantum states on the two inputs) to the resulting increase in the progress measure for the transcript distributions. Iterating this inequality yields both lower bounds through the same general argument.

The AND-OR lower bound requires a more intricate analysis. We combine the transcript-based measure with ideas from the hybrid compressed-oracle framework of Hamoudi et al. \cite{hamoudi2024nisq}, including a reduction to an alternative model of hybrid computation and separate progress measures tracking the classical and quantum information accumulated by the algorithm. Unlike in their work, however, we incorporate these ingredients into an adversary-style argument rather than using them to track a compressed oracle. This combination may also be useful for other problems to which the compressed-oracle method does not directly apply.

\section{Computational model}

{\bf Hybrid algorithms.}
We consider \emph{hybrid quantum--classical algorithms} in the query model. 
The computation is classical, but it may invoke quantum subroutines, each of which is allowed to make at most $q$ oracle queries. At the end of each quantum subroutine,
all quantum registers are measured and the measurement results are passed to the classical algorithm.

Formally, such an algorithm is represented by a finite rooted tree. 
Each internal vertex $v$ has a corresponding quantum algorithm $A_v$. The algorithm $A_v$ makes $r_v\leq q$ queries to the oracle and ends with a measurement. The outgoing edges from $v$ are labelled by the possible measurement outcomes of $A_v$.
Each leaf of the tree carries a classical output value that represents the final result of the computation.

The computation proceeds as follows.
Execution begins at the root of the tree.
At a vertex $v$, the algorithm executes the quantum subroutine $A_v$ and observes its measurement outcome.
The computation then follows the edge corresponding to this outcome to the appropriate child vertex.
This process continues until a leaf is reached, at which point 
its label is output as a result of the computation.

Let $T$ be the maximum number of quantum subroutines on a root-to-leaf path. We may pad shorter paths with zero-query dummy subroutines, so that every execution has $T$ stages. 

The complexity of the algorithm is
\[
L=\max_P\sum_{v\in P}r_v,
\]
where the sum is over the internal vertices $v$ on $P$, and the maximum is over root-to-leaf computation paths $P$ that can occur on a valid input.

{\bf Computational problems.}
We consider computational problems of several types. The first type are Boolean functions $F(x_1, \ldots, x_n)$. Then, an oracle call applies the standard phase
oracle $O_x\ket{i, v, w} = (-1)^{v x_i} \ket{i, v, w}$ where $i\in [n]$ is the index to be queried, $v\in\{0, 1\}$ controls whether a query is applied and $w$ denotes the algorithm's registers that do not affect the query.

The second type are problems where the input is an unknown unitary $U$ and we have to determine if $U$ has a certain property or estimate some quantity that depends on $U$.
Then, an oracle call is an application of controlled-$U$ or controlled-$U^{\dagger}$.

The hybrid query complexity of a computational problem $F$, denoted $Q(F;q)$, is the smallest complexity of a hybrid quantum-classical algorithm that solves $F$ with 
error at most $1/3$ on every valid input, using quantum subroutines with at most $q$ queries.

We also use $Q(F)$ to denote the unrestricted bounded-error quantum query complexity.

\section{Transcripts and progress measure}
\label{sec:general}

In this section, we describe the framework for the proofs of Theorems \ref{thm:classification} and \ref{thm:2} and prove a key lemma which serves as a basis for proofs of both theorems.
Given this lemma, both theorems can be deduced by quite short arguments.

\textbf{Transcripts and the global progress measure.}
Let $y_t$ be the measurement outcome of the $t$-th quantum subroutine executed, and let $Y_t=(y_1,\ldots,y_t)$. We refer to $Y_t$ as the (classical) transcript of the hybrid algorithm. Let
\[
p_h := \Pr[Y_t = h \mid U_0],\qquad
q_h := \Pr[Y_t = h \mid U_1].
\]
If the oracle is $U_0$ with probability 1/2 and $U_1$ with probability 1/2, the overall probability 
of $Y_t = h$ is $\frac{p_h + q_h}{2}$.

We write $p^{(t)}=(p_h)_h$ and $q^{(t)}=(q_h)_h$ for the two distributions of $Y_t$. For any two probability distributions $p,q$ on the same finite set, we define
\begin{equation}
\label{eq:defM_simple}
M(p,q):=\sum_z\frac{(p(z)-q(z))^2}{2(p(z)+q(z))},
\end{equation}
where a summand is taken to be zero if $p(z)=q(z)=0$. We use
$ M_t:=M\bigl(p^{(t)},q^{(t)}\bigr)$ as our measure of progress.

If we choose the oracle uniformly at random from $\{U_0, U_1\}$ and run the algorithm for $t$ blocks, then $M_t$ is equal to the expectation of
$\bigl(\frac{p_{h} - q_{h}}{p_{h} + q_h}\bigr)^2$.

A standard Cauchy--Schwarz bound relates $M_t$ to the total variation distance:
\[
\mathrm{TV}(p^{(t)}, q^{(t)})
= \frac{1}{2} \sum_h |p_h - q_h|
\le \sqrt{M_t}.
\]
The final output is a function of the final transcript $Y_T$. Therefore, the total variation distance between the transcript distributions is at least the total variation distance between the corresponding output distributions. In particular, if the algorithm distinguishes $U_0$ from $U_1$ with error at most $1/3$, then
$\operatorname{TV}\bigl(p^{(T)},q^{(T)}\bigr)\ge \frac13$,
and hence $M_T\ge 1/9$.

\medskip
\noindent\textbf{Conditioning on a history.}
Fix integers $t \le j$ and a transcript $h$ of length $t$ with $p_h + q_h > 0$. 
Let $h'$ be a length-$j$ transcript and write $h \vdash h'$ if $h$ is a prefix of $h'$.
Assuming that the oracle $U$ is chosen to be $U_0$ with probability 1/2 and $U_1$ with probability 1/2, we define the conditional probabilities
\[
p_{h'}(h) := \Pr[Y_j = h',\,U = U_0 \mid Y_t = h], \mbox{~}
q_{h'}(h) := \Pr[Y_j = h',\,U = U_1 \mid Y_t = h].
\]
If $h \not\vdash h'$, then $p_{h'}(h) = q_{h'}(h) = 0$. If $h \vdash h'$, then because of the uniform
distribution on $U$ we have
\[
p_{h'}(h) = \frac{p_{h'}}{p_h + q_h},\qquad
q_{h'}(h) = \frac{q_{h'}}{p_h + q_h}.
\]
In particular, this means
\begin{equation}
\label{eq:equal}
\frac{p_{h'}(h) - q_{h'}(h)}{p_{h'}(h) + q_{h'}(h)} =
\frac{p_{h'} - q_{h'}}{p_{h'} + q_{h'}}.
\end{equation}

We define the conditional progress measure
\[
M_j(h) := \mathbb{E}_{h':h \vdash h'}\left[\left(
\frac{p_{h'}(h) - q_{h'}(h)}{p_{h'}(h) + q_{h'}(h)}
\right)^{\!\!2}\right],
\]
where the expectation is with respect to the probability distribution $\Pr[Y_j = h' \mid Y_t = h]$.
Then the (unconditional) progress measure satisfies 
\[
M_j = \mathbb{E}_{h'}\left[\left(\frac{p_{h'} - q_{h'}}{p_{h'} + q_{h'}}\right)^{\!\!2}\right] = \sum_h \Pr[Y_t = h] M_j(h) = \sum_h \frac{p_h + q_h}{2} M_j(h),
\]
with the second equality following from (\ref{eq:equal}).
Consequently, to upper bound $M_{t + 1} - M_t$ it suffices to upper bound
$M_{t + 1}(h) - M_t(h)$ for every $h$ and then average with weights $\tfrac{p_h + q_h}{2}$.
The following bound retains a factor $1 - M_t(h)$ that will also be needed in the search argument.

\begin{lemma}
\label{lem:cond}
Let $h$ be a transcript of length $t < T$ with $p_h + q_h > 0$, and consider the subroutine selected after $h$. Let $\ket{\psi_0}$ and $\ket{\psi_1}$ be its states immediately before the final measurement when run with $U_0$ and $U_1$, respectively. Then
\begin{equation}
M_{t + 1}(h) - M_t(h) \leq (1 - M_t(h)) \nrm{\ket{\psi_0} - \ket{\psi_1}}^2.
\label{eq:increment_final}
\end{equation}
\end{lemma}

\begin{proof}
Fix a history $h$ with $p_h+q_h>0$, and let
\[
a=\frac{p_h}{p_h+q_h}
\]
be the posterior probability that the oracle is $U_0$. Thus
\[
M_t(h)=(2a-1)^2,
\qquad
1-M_t(h)=4a(1-a).
\]
If $a\in\{0,1\}$, the oracle is already determined by $h$, so $M_{t+1}(h)=M_t(h)=1$ and the claimed inequality holds. We therefore assume $0<a<1$.

Write $|\psi_0\rangle$ and $|\psi_1\rangle$ for the normalized states immediately before the next block’s final measurement under the two oracles, with $h$ fixed. Let $r_i$ and $s_i$ be their respective probabilities of measurement outcome $i$. Conditioned on $h$, this outcome has probability
\[
w_i=ar_i+(1-a)s_i.
\]
By the definition of conditional progress,
\[
M_{t+1}(h)
=\sum_i\frac{\bigl(ar_i-(1-a)s_i\bigr)^2}{w_i}
=1-4a(1-a)\sum_i\frac{r_i s_i}{w_i},
\]
where outcomes with $w_i=0$ are omitted.

Let $\{E_i\}$ be the POVM describing the common final measurement. Applying Cauchy–Schwarz first to each measurement operator and then to the sum over outcomes gives
\[
|\langle\psi_0|\psi_1\rangle|
=\left|\sum_i\langle\psi_0|E_i|\psi_1\rangle\right|\\
\le \sum_i\sqrt{r_i s_i}\\
=\sum_i\sqrt{w_i}\sqrt{\frac{r_i s_i}{w_i}}\\
\le \left(\sum_i\frac{r_i s_i}{w_i}\right)^{1/2},
\]
using $\sum_i w_i=1$. Consequently,
\[
1-M_{t+1}(h)
\ge
\bigl(1-M_t(h)\bigr)|\langle\psi_0|\psi_1\rangle|^2.
\]
Finally, since both states are normalized,
\[
1-|\langle\psi_0|\psi_1\rangle|^2
\le 2-2\operatorname{Re}\langle\psi_0|\psi_1\rangle
=\bigl\||\psi_0\rangle-|\psi_1\rangle\bigr\|^2.
\]
Combining these inequalities yields
\[
M_{t+1}(h)-M_t(h)
\le
\bigl(1-M_t(h)\bigr)
\bigl\||\psi_0\rangle-|\psi_1\rangle\bigr\|^2,
\]
as required. 
\end{proof}

\section{Proofs of Theorems \ref{thm:classification} and \ref{thm:2}}

\classification*

\begin{proof}
For each quantum subroutine, we have

\begin{lemma}
\label{lem:hybrid}
For any $r$-query quantum algorithm (with all non-query operations independent of the oracle),
\[
\bigl\|\ket{\psi_0} - \ket{\psi_1}\bigr\| \le r\bigl\|U_0 - U_1\bigr\|
\le r \epsilon .
\]
\end{lemma}

The lemma follows by a standard hybrid argument in which we replace the $r$ oracle calls one at a time and apply the triangle inequality. Each replacement changes the state by at most $\|U_0-U_1\|$.

Let $L$ be the total number of queries in the worst-case and $r_h\le q$ be the number of queries in the block following $h$.
Applying \Cref{lem:hybrid} with $r_h$ in place of $r$, and using $1 - M_t(h) \leq 1$ in \Cref{lem:cond}, gives $M_{t + 1}(h) - M_t(h) \leq (r_h \epsilon)^2$. Averaging over histories of length $t$ gives
\[
M_{t + 1} - M_t \leq \epsilon^2 \sum_{|h| = t} \frac{p_h + q_h}{2} r_h^2.
\]
Under the uniform prior, $\sum_{t < T} \sum_{|h| = t} (p_h + q_h)r_h/2$ is the expected total number of queries and is bounded by $L$. Since $M_0 = 0$ and $r_h^2 \leq qr_h$, summing the preceding inequality gives
\[
M_T \leq \epsilon^2 \sum_{t < T} \sum_{|h| = t} \frac{p_h + q_h}{2} r_h^2 \leq q \epsilon^2 L.
\]
Error at most $1/3$ requires $M_T\ge 1/9$, and hence
\[
L=\Omega\!\left(\frac{1}{q\epsilon^2}\right).
\]
Since every hybrid algorithm is also an unrestricted quantum query algorithm, the standard hybrid argument applied to the entire computation gives $L=\Omega(1/\epsilon)$. Combining the two bounds proves the theorem.\end{proof}

\grover*

\begin{proof}
We use the complement $M' = 1 - M$, with conditional values $M'_j(h) = 1 - M_j(h)$. By \eqref{eq:increment_final},
\[
M'_{t + 1}(h) \geq \left(1 - \nrm{\ket{\psi_0} - \ket{\psi_1}}^2\right)M'_t(h).
\]
For search, let $p^{(i)}_h$ be the probability of transcript $h$ on the input whose unique $1$ is at position $i$, and let $q_h$ be its probability on the all-zero input. Write $p^{(i)}$ and $q$ for these distributions at time $t$, and put
\[
M_t^{\prime(i)} = 1 - M(p^{(i)}, q) = \sum_{|h| = t} \frac{2 p^{(i)}_h q_h}{p^{(i)}_h + q_h}.
\]
We use $\sum_{i = 1}^N M_t^{\prime(i)}$ as the progress measure. Initially $M_0^{\prime(i)} = 1$. If the error is at most $1/3$, then $\mathrm{TV}(p^{(i)}, q) \geq 1/3$ at termination, so $M_T^{\prime(i)} \leq 8/9$.

Similarly to the proof of Theorem \ref{thm:classification}, let $L$ be the total number of queries in the worst-case.
Fix $t < T$ and a history $h$ of length $t$, and consider the subroutine selected after $h$, which makes $r_h$ queries. Let $\ket{\psi_i^{(h)}}$ and $\ket{\psi_0^{(h)}}$ be its final states before measurement on the input with unique $1$ at $i$ and on the all-zero input, respectively. The argument in \cite[Section~3.1]{andor-bbbv} gives
\[
\sum_{i = 1}^N \nrm{\ket{\psi_i^{(h)}} - \ket{\psi_0^{(h)}}}^2 \leq 4 r_h^2.
\]
For each $i$, apply \Cref{lem:cond} to the corresponding input pair and average over histories $h$ with weights $(p_h^{(i)} + q_h)/2$. Summing over $i$ gives
\[
\sum_i \left(M_t^{\prime(i)} - M_{t + 1}^{\prime(i)}\right) \leq \sum_i \sum_{|h| = t} \frac{2 p^{(i)}_h q_h}{p^{(i)}_h + q_h} \nrm{\ket{\psi_i^{(h)}} - \ket{\psi_0^{(h)}}}^2.
\]
Using $2 p^{(i)}_h q_h/(p^{(i)}_h + q_h) \leq 2 q_h$ and the bound on the sum of squared distances, we obtain
\[
\sum_i \left(M_t^{\prime(i)} - M_{t + 1}^{\prime(i)}\right) \leq 8 \sum_{|h| = t} q_h r_h^2.
\]
The sum $\sum_{t < T} \sum_{|h| = t} q_h r_h$ is the expected query count on the all-zero input and is at most $L$. Summing over $t$ and using $r_h^2 \leq qr_h$ gives
\[
\frac{N}{9} \leq \sum_i \left(M_0^{\prime(i)} - M_T^{\prime(i)}\right) \leq 8 \sum_{t < T} \sum_{|h| = t} q_h r_h^2 \leq 8 qL.
\]
Thus $L = \Omega(N/q)$. The standard unrestricted quantum lower bound $\Omega(\sqrt N)$ \cite{andor-bbbv} completes the proof.
\end{proof}

\section{Two-level AND-OR trees}
\label{sec:andor-noisy}

Consider the function
\[
F(x) = \bigwedge_{i = 1}^{m} \bigvee_{j = 1}^{n} x_{ij},
\qquad x \in \{0, 1\}^{m \times n},
\]
where $m, n \geq 2$ are integers. The $n$ bits with first index $i$ form OR block $i$. 

We tightly characterize the complexity of $F$ in the hybrid query model.

\andor*

\begin{proof}
For the upper bound, every subformula of depth at most two on at most $s$ leaves has a quantum algorithm of error at most $1/3$ using $\bigO(\sqrt{s})$ queries \cite[Section~4]{andor-hmw}. For $q$ above an absolute constant, choose an integer $s = \Theta(\min\{q^2, mn\})$ small enough that these evaluations use at most $q$ queries. If $s \leq n$, partition each OR block into groups of $s$ bits, with a possibly smaller final group. If $s > n$, group $\lfloor s/n\rfloor$ complete OR blocks under each AND subformula, again allowing a smaller final group. Their values form an AND-OR formula of depth at most two and size $\bigO(mn/s)$. The noisy classical formula-evaluation result of \cite[Section~6, p.~1017]{andor-frpu} combines these values with $\bigO(mn/s)$ requests, each having independent error at most $1/3$. Implement each request by a fresh quantum evaluation, measuring and discarding all qubits afterwards. This gives $\bigO((mn/s) \sqrt{s}) = \bigO(\sqrt{mn} + mn/q)$ total queries. For smaller $q$, reading every bit classically gives the same asymptotic bound.

We now proceed to proving the lower bound. We use the same sets of difficult inputs as in the standard proof of the lower bound for AND-OR in \cite[Section~5.3]{andor-ambainis}.

Let $B$ consist of the inputs with exactly one $1$ in each OR block. Let $A$ consist of the inputs with one all-zero OR block and exactly one $1$ in every other block. Thus $F$ is $1$ on $B$ and $0$ on $A$. We will show the lower bound for computing $F$ if the input is restricted to being in $A \cup B$.

For $y \in B$, let $y^{(i)}$ be obtained by changing the unique $1$ in block $i$ to $0$. Let $R$ be the set consisting of all pairs $(y^{(i)}, y)$.
We have $|B| = n^m$ and $|R| = m n^m$.
We will show that, if the algorithm makes too few queries, then for some choice of $y$ and $i$ it behaves similarly on $y$ and $y^{(i)}$ and, therefore, is incorrect on one of these two inputs.

\medskip
\noindent\textbf{The noisy oracle.}
A key idea is considering a different model of hybrid computation (from \cite{hamoudi2024nisq}) in which, instead of a limit $q$ on the number of queries in one quantum subroutine, there is a noise with some probability in every query.

The standard quantum query is
\[
O_x \ket{i, j, \nu, w} = (-1)^{\nu x_{ij}} \ket{i, j, \nu, w}
\]
where $(i, j) \in [m] \times [n]$ is the index to be queried, $\nu \in \{0, 1\}$ is a phase bit, and $w$ is the workspace.
For a noisy query, we fix a noise probability $0 < \rho \leq 1$. A call of the noisy oracle applies $O_x$ with probability $1 - \rho$. With probability $\rho$, it measures the index and phase registers in their basis and returns the bit $x_{ij}$ at the measured index. Both alternatives append a flag; the classical alternative also appends the queried value. 

This can be described more formally as follows.
Let $P_{i, j, \nu} = \ket{i, j, \nu} \bra{i, j, \nu} \otimes I$ and $P_i = \sum_{j, \nu} P_{i, j, \nu}$.
On a density matrix $\sigma$, the noisy oracle applies the map
\begin{equation}
\mathcal N_{\rho, x}(\sigma) = (1 - \rho)O_x \sigma O_x^* \otimes \ket{\mathrm Q, \bot} \bra{\mathrm Q, \bot} + \rho \sum_{i, j, \nu} P_{i, j, \nu} \sigma P_{i, j, \nu} \otimes \ket{\mathrm C, x_{ij}} \bra{\mathrm C, x_{ij}}.
\label{eq:andor-channel}
\end{equation}
The appended registers are a flag register and a value register, with bases $\{\ket{\mathrm Q}, \ket{\mathrm C}\}$ and $\{\ket{\bot}, \ket{0}, \ket{1}\}$. 

Let $T$ be the maximum number of calls made by the algorithm on inputs in $A \cup B$, and let $\Sigma$ be its finite output alphabet. On input $x$, let $p_x(o)$ be the probability of output $o \in \Sigma$, and let $C(x)$ be the expected number of calls. We set
\[
C_+ = \frac{1}{|B|} \sum_{y \in B} C(y).
\]

Let $M(p, q)$ be the measure in \eqref{eq:defM_simple} applied to distributions $p, q$ on outputs $\Sigma$, with a zero summand when both probabilities vanish. The main part of our proof is

\begin{lemma}
\label{thm:andor-noisy}
For every algorithm using \eqref{eq:andor-channel} that terminates after at most $T < \infty$ calls on every input in $A \cup B$,
\begin{equation}
\frac{1}{|R|} \sum_{(x, y) \in R} M(p_x, p_y)
\leq \frac{4096}{\rho mn} C_+.
\label{eq:andor-main}
\end{equation}
Consequently, if an accepting-output set $E \subseteq \Sigma$ satisfies $p_y(E) \geq 1 - \epsilon$ for every $y \in B$ and $p_x(E) \leq \epsilon$ for every $x \in A$, where $0 \leq \epsilon < 1/2$, then
\[
C_+ \geq \frac{(1 - 2 \epsilon)^2}{4096} \rho mn.
\]
The same lower bound holds for the maximum number of calls $T$.
\end{lemma}

To obtain the lower bound part of \Cref{thm:andor-main}, we take $\rho = 1/(2 q)$ and apply the restart construction of \cite[Proposition~3.4]{hamoudi2024nisq}: restart a subroutine whenever a classical-query flag occurs. An $r$-query circuit, $r \leq q$, has a noise-free attempt with probability $(1 - \rho)^r \geq 1/2$; conditional on this event its output distribution is unchanged, and the expected simulation cost is at most $2 r$. Cap the number of attempts at each node so that the total abort probability is at most $1/12$. On abort, output an arbitrary fixed value. The original error $1/3$ then increases by at most the abort probability $1/12$, giving error at most $5/12$.

An algorithm using at most $L$ queries therefore gives a noisy algorithm with bounded query count, error at most $5/12$, and $C_+ \leq 2 L$. \Cref{thm:andor-noisy} yields $L = \Omega(mn/q)$. The usual quantum lower bound $Q(F) = \Omega(\sqrt{mn})$ follows from \cite[Theorem~2 and Section~5.3]{andor-ambainis}.
\end{proof}

We first compare the algorithm on an input with one all-zero OR block to the same algorithm on $n$ inputs obtained by adding a $1$ in that block.

Fix a block $i$ and the position of the unique $1$ in every other block. Let $x_0 \in A$ have block $i$ all zero, and let $x_j \in B$ be obtained from $x_0$ by setting its bit $(i, j)$ to $1$, for $j \in [n]$. The algorithm can query all blocks throughout this comparison.

For a query state $\sigma$, the query weight on block $i$ is $\operatorname{Tr}(P_i \sigma)$, the sum of the query magnitudes on its indices \cite[Definition~3.2]{andor-bbbv}. Let $C_i(x)$ be the expected sum of these weights over the execution on input $x$. Since $\sum_i P_i = I$, we have $\sum_i C_i(x) = C(x)$. We will prove

\begin{lemma}
\label{lem:one-block}
\begin{equation}
\frac{1}{n} \sum_{j = 1}^n M(p_{x_0}, p_{x_j})
\leq \frac{4096}{\rho n} \frac{1}{n} \sum_{j = 1}^n C_i(x_j).
\label{eq:andor-one-block}
\end{equation}
\end{lemma}

We now deduce \Cref{thm:andor-noisy} from \Cref{lem:one-block}.
\begin{proof}[Proof of \Cref{thm:andor-noisy}]

We average
\eqref{eq:andor-one-block} over the $mn^{m - 1}$ choices of block $i$ and positions of the unique $1$s outside it. As these choices and $j \in [n]$ vary, $(x_0, x_j)$ runs once through $R$, and $(x_j, i)$ runs once through $B \times [m]$. The average on the left is therefore the relation average in \eqref{eq:andor-main}; the cost average is
\[
\frac{1}{mn^m} \sum_{y \in B} \sum_{i = 1}^m C_i(y)
= \frac{1}{mn^m} \sum_{y \in B} C(y)
= \frac{C_+}{m}.
\]
This proves \eqref{eq:andor-main}. The decision conclusion follows from the bound $\mathrm{TV}(p, q)^2 \leq M(p, q)$ established in \Cref{sec:general}, since $\mathrm{TV}(p_x, p_y) \geq 1 - 2 \epsilon$ and $C_+ \leq T$.
\end{proof}

\section{Comparison of inputs in one OR block}
\label{sec:andor-comparison}

In this section, we prove \Cref{lem:one-block}.

\medskip
\noindent\textbf{Query and output vectors.}
We fix the inputs $x_0, \ldots, x_n$ from \Cref{lem:one-block}. Let $\ket{\psi_\inlabel}$ be the algorithm's starting state (which is independent of the input). We purify the quantum algorithm in the standard way, postponing all measurements to the end of the algorithm and replacing transformations that depend on classical outcomes by controlled transformations that depend on qubits. The classical outputs of a noisy query are copied into separate quantum registers, which can be used as controls for transformations but are not modified.

A history $h$ records the successive oracle outcomes: $Q$ for a quantum query, and $(C, i', j', \nu, v)$ for a classical query at $(i', j')$, with phase bit $\nu$ and returned value $v$. Write $|h|$ for its length and $he$ for the history obtained by appending outcome $e$. Other measurement records are included in the purification rather than in $h$.

For input $x_j$ and a history $h$ of length $t$, let $|\psi_h^j\rangle$ be the subnormalized vector for the part of the computation about to make its $(t + 1)$st query, and let $|\psi_{h, \mathrm{out}}^j\rangle$ be the subnormalized vector for the part that terminates after these $t$ queries. The latter includes the output register. Their squared norms are the corresponding joint probabilities, including the probability of the history $h$.

Each oracle call preserves total squared norm when summed over its outcomes. The intervening input-independent operations preserve inner products and split the state into orthogonal query and output parts. Terminated parts are retained unchanged, with distinct histories and termination times occupying orthogonal subspaces.

All branches terminate after at most $T$ calls. Reading the output register in the direct sum of the terminal vectors gives the distribution $p_{x_j}$. Consequently,
\begin{equation}
\label{eq:new}
\sum_{|h| \le T} \|\psi_{h, \mathrm{out}}^j\|^2 = 1,
\qquad
C_i(x_j) =
\sum_{|h| < T} \|P_i|\psi_h^j\rangle\|^2.
\end{equation}

To compare these inputs in one vector, adjoin an input-label register \cite[Section~3]{andor-ambainis} with space $\mathcal Z = \operatorname{span}\{\ket{0}, \ldots, \ket{n}\}$. Label $0$ represents the all-zero block, and label $j \in [n]$ represents the block whose unique $1$ is at position $j$. The algorithm's unitaries act as identity on this register. Define
\begin{equation}
\ket{\phi_h} = \frac{1}{\sqrt{2}} \ket{0} \ket{\psi_h^0} + \frac{1}{\sqrt{2 n}} \sum_{j = 1}^n \ket{j} \ket{\psi_h^j}.
\label{eq:andor-comparison-vectors}
\end{equation}
The coefficients assign weight $1/2$ to $x_0$ and weight $1/(2 n)$ to each $x_j$. Use the same coefficients to define $|\phi_{h, \mathrm{out}}\rangle$ from the output vectors. Write $|\widehat\phi_{he}\rangle$ for the subnormalized comparison vector immediately after the oracle outcome $e$, before the subsequent input-independent operations.

\medskip
\noindent\textbf{The progress measure.}
We keep track of positions excluded by classical answers and differences between the remaining input vectors. If $h$ has not recorded a classical answer $1$ in block $i$, the positions not yet excluded are
\[
S_h = [n] \setminus \{s:(\mathrm C, i, s, \nu, 0) \text{ occurs in }h \text{ for some } \nu\}.
\]
If $j \notin S_h$, then $h$ contains a classical answer $0$ at $(i, j)$, which has probability zero on input $x_j$. Hence
\begin{equation}
\ket{\psi_h^j} = \ket{\psi_{h, \outlabel}^j} = 0\qquad\text{for }j \in [n] \setminus S_h.
\label{eq:andor-support}
\end{equation}
For $S \subseteq [n]$, let $\Pi_S = \sum_{j \in S} \ket{j} \bra{j}$ project onto the labels in $S$. For nonempty $S$, set
\[
\lambda_S = \frac{n - |S|}{n - 1},\qquad
\ket{v_S} = \frac{\sqrt{n} \ket{0} + \sum_{j \in S} \ket{j}}{\sqrt{n + |S|}},\qquad
\Pi_{v_S}^{\perp} = I_{\mathcal Z} - \ket{v_S} \bra{v_S}.
\]
The weights in $\ket{v_S}$ come from \eqref{eq:andor-comparison-vectors}: a comparison vector supported on $\{0\} \cup S$ has $\Pi_{v_S}^{\perp} \ket{\phi} = 0$ if and only if its subnormalized algorithm vectors for $x_0$ and $x_j$, $j \in S$, agree. Thus the perpendicular component measures differences between those vectors. The coefficient $\lambda_S$ increases as classical answers exclude positions. We combine these quantities in the raw progress
\[
\Psi_S^{\raw}(\ket{\phi})
= \lambda_S \nrm{\phi}^2 + 4(1 - \lambda_S) \nrm{\Pi_{v_S}^{\perp} \ket{\phi}}^2.
\]
\medskip
\noindent\textbf{Capping the progress.}
For a history with no recorded $1$ in block $i$ and $|S_h| \geq 2$, define its progress contribution by
\[
\Psi_h(\ket{\phi})
= \min\left\{\Psi_{S_h}^{\raw}(\ket{\phi}), \frac{1}{4} \nrm{\phi}^2\right\}.
\]
In all other cases set $\Psi_h(\ket{\phi}) = \nrm{\phi}^2/4$. Call a vector saturated when its contribution equals the cap, and unsaturated otherwise. Conservation of squared norm ensures that the sum of outcome contributions from a saturated query vector cannot exceed its current contribution. Apply the cap separately at each history, to its query vector and to the direct sum of its output vectors.

Because input-independent unitaries act trivially on $\mathcal Z$, they preserve the raw form and $\nrm{\cdot}^2/4$. Both forms are additive under the query/output split. The inequality $\min\{a_1, d_1\} + \min\{a_2, d_2\} \leq \min\{a_1 + a_2, d_1 + d_2\}$ therefore gives
\begin{equation}
\Psi_{he}(\ket{\phi_{he}}) + \Psi_{he}(\ket{\phi_{he, \outlabel}})
\leq \Psi_{he}(\ket{\widehat\phi_{he}}).
\label{eq:andor-processing}
\end{equation}
The same statement holds for the initial split. At query count $t$, the total progress is
\[
\Psi_t = \sum_{|h| \leq t} \Psi_h(\ket{\phi_{h, \outlabel}}) + \sum_{|h| = t} \Psi_h(\ket{\phi_h}).
\]
At the final query count $T$, the query vectors vanish, so $\Psi_T$ is the sum of the output contributions.

\begin{lemma}
\label{lem:andor-progress}
Let $x_0 \in A$ have block $i$ all zero. For $j \in [n]$, let $x_j \in B$ be obtained from $x_0$ by changing bit $(i, j)$ to $1$. Then
\begin{align}
\Psi_0 &= 0, \label{eq:andor-initial-progress}\\
\frac{1}{n} \sum_{j = 1}^n M(p_{x_0}, p_{x_j})
&\leq 16 \Psi_T, \label{eq:andor-final-progress}\\
\Psi_{t + 1} - \Psi_t
&\leq \frac{256}{\rho n^2} \sum_{j = 1}^n \sum_{h: |h| = t} \nrm{P_i \ket{\psi_h^j}}^2.
\label{eq:andor-query-increment}
\end{align}
The last inequality holds for $0 \leq t < T$.
\end{lemma}

The initial comparison vector is $\ket{v_{[n]}} \ket{\psi_{\inlabel}}$, whose raw progress is zero. Nonnegativity and \eqref{eq:andor-processing} therefore give \eqref{eq:andor-initial-progress}. We prove \eqref{eq:andor-query-increment} in \Cref{sec:andor-query} and \eqref{eq:andor-final-progress} in \Cref{sec:andor-final}. We now deduce \Cref{lem:one-block}, given those equations.

Summing \eqref{eq:andor-query-increment} and using $\Psi_0 = 0$, \eqref{eq:new}, and \eqref{eq:andor-final-progress} gives
\[
\frac{1}{n} \sum_{j = 1}^n M(p_{x_0}, p_{x_j})
\leq 16 \Psi_T
\leq \frac{4096}{\rho n^2} \sum_{j = 1}^n C_i(x_j).
\]
This proves \Cref{lem:one-block}.

\section{Proof of Lemma \ref{lem:andor-progress} }

\subsection{Progress under one noisy query}
\label{sec:andor-query}

We prove \eqref{eq:andor-query-increment} by first bounding the contribution from one query index and phase bit, then summing over them.

Fix $0 \leq t < T$ and $h: |h| = t$. If $\ket{\phi_h}$ is saturated, conservation of squared norm gives
\begin{equation}
\sum_{e} \Psi_{he}(\ket{\widehat\phi_{he}})
\leq \frac{1}{4} \sum_{e} \nrm{\widehat\phi_{he}}^2
= \Psi_h(\ket{\phi_h}).
\label{eq:andor-saturated}
\end{equation}
The sums range over the possible outcomes $e$ of the noisy oracle.

Now suppose $\ket{\phi_h}$ is unsaturated. Put $S = S_h$, $r = |S|$, and $\lambda = \lambda_S$. Then $r \geq 2$. Since the raw form is at least $\lambda \nrm{\phi_h}^2$, unsaturation implies
\begin{equation}
\lambda < \frac{1}{4},\qquad r > \frac{3 n}{4},\qquad \frac{r}{n + r} > \frac 2 5.
\label{eq:andor-set-size}
\end{equation}

For $s \in S$ and $\nu \in \{0, 1\}$, take the query component $\ket{\phi_{s, \nu}} = (I_{\mathcal Z} \otimes P_{i, s, \nu}) \ket{\phi_h}$ and put $S' = S \setminus \{s\}$. A quantum query applies the phase $(-1)^\nu$ to input label $s$. A classical query separates that label from the others: answer $1$ identifies $x_s$, while answer $0$ replaces $S$ by $S'$. Suppressing the factors $\sqrt{1 - \rho}, \sqrt{\rho}$ and the appended flag and value kets, the postquery vectors are
\[
\ket{\widehat\phi_{\mathrm Q}^{\nu}} = (I - 2 \nu \Pi_{\{s\}}) \ket{\phi_{s, \nu}},\quad\ket{\widehat\phi_{\mathrm C, 1}} = \Pi_{\{s\}} \ket{\phi_{s, \nu}},\quad\ket{\widehat\phi_{\mathrm C, 0}} = (I - \Pi_{\{s\}}) \ket{\phi_{s, \nu}}.
\]
For answer $1$, the contribution equals the cap; for answer $0$, it is the minimum of the raw form for $S'$ and the cap. Measure the local changes relative to the prequery raw form:
\begin{align}
\Delta_{\mathrm Q}^{\nu} &:= \Psi_S^{\raw}(\ket{\widehat\phi_{\mathrm Q}^{\nu}}) - \Psi_S^{\raw}(\ket{\phi_{s, \nu}}), \nonumber\\
\Delta_{\mathrm C} &:= \frac 1 4 \nrm{\widehat\phi_{\mathrm C, 1}}^2 + \min\left\{\Psi_{S'}^{\raw}(\ket{\widehat\phi_{\mathrm C, 0}}), \frac 1 4 \nrm{\widehat\phi_{\mathrm C, 0}}^2\right\} - \Psi_S^{\raw}(\ket{\phi_{s, \nu}}).
\label{eq:andor-increment-definitions}
\end{align}
The quantum outcome has a single history, so we use its raw form before summing over query components. Classical outcomes have separate histories and use their capped contributions. Weighting the changes by the query probabilities gives
\[
\Delta_{s, \nu} := (1 - \rho) \Delta_{\mathrm Q}^{\nu} + \rho \Delta_{\mathrm C}.
\]
When $|S'| = 1$, its raw form is the squared norm, so the minimum still gives the cap.

\begin{lemma}
\label{lem:andor-position}
Let $n \geq 2$, $0 < \rho \leq 1$, and $S \subseteq [n]$ satisfy $|S| \geq 2$ and $\lambda_S < 1/4$. Let $s \in S$ and $\nu \in \{0, 1\}$. For any vector $\ket{\phi_{s, \nu}}$ with input-label support in $\{0\} \cup S$, the increments defined above satisfy
\[
\Delta_{s, \nu} \leq \frac{512}{\rho n} \nrm{\Pi_S \ket{\phi_{s, \nu}}}^2.
\]
\end{lemma}

\begin{proof}
Write $\ket{\phi} = \ket{\phi_{s, \nu}}$, $S' = S \setminus \{s\}$, $r = |S|$, and $\lambda = \lambda_S$. The assumption on $S$ gives \eqref{eq:andor-set-size}. We use the components on label $s$ and in the common direction for $S'$:
\[
\ket{a} = (\bra{s} \otimes I) \ket{\phi},\qquad\ket{b} = (\bra{v_{S'}} \otimes I) \ket{\widehat\phi_{\mathrm C, 0}}.
\]
Since $\ket{v_S} = (\ket{s} + \sqrt{n + r - 1} \ket{v_{S'}})/\sqrt{n + r}$, we have
\begin{equation}
\nrm{\Pi_{v_S}^{\perp} \ket{\phi}}^2 = \nrm{\widehat\phi_{\mathrm C, 0}}^2 - \nrm{b}^2 + \frac{1}{n + r}\left\lVert\sqrt{n + r - 1} \ket{a} - \ket{b}\right\rVert^2.
\label{eq:andor-local-projection}
\end{equation}
Replacing $\ket{a}$ by $(-1)^\nu \ket{a}$ in \eqref{eq:andor-local-projection} gives the quantum change. For classical answer $0$, use the raw form and $\lambda_{S'} = \lambda + 1/(n - 1)$. This gives
\begin{align}
\Delta_{\mathrm Q}^{\nu} &= 16 \nu(1 - \lambda) \frac{\sqrt{n + r - 1}}{n + r} \Re\langle a \mid b\rangle, \nonumber\\
\Delta_{\mathrm C} &\leq (\tfrac 1 4 - \lambda) \nrm{a}^2 + \frac{4 \nrm{b}^2 - 3 \nrm{\widehat\phi_{\mathrm C, 0}}^2}{n - 1} - \frac{4(1 - \lambda)}{n + r}\left\lVert\sqrt{n + r - 1} \ket{a} - \ket{b}\right\rVert^2.
\label{eq:andor-position-changes}
\end{align}
Thus $\Delta_{\mathrm Q}^{\nu} \leq 16 \nrm{a} \nrm{b}/\sqrt n$. After expanding the squared norm in the classical bound, $n + r \geq 4$ makes the coefficient of $\nrm{a}^2$ at most $\frac 1 4 - \lambda - 3(1 - \lambda) < -2$. Using $\nrm{b} \leq \nrm{\widehat\phi_{\mathrm C, 0}}$ and $1/(n - 1) \leq 2/n$ for the other terms gives
\begin{equation}
\Delta_{\mathrm C} \leq -2 \nrm{a}^2 + \frac{8}{\sqrt n} \nrm{a} \nrm{b} + \frac{2}{n} \nrm{b}^2.
\label{eq:andor-classical-decrease}
\end{equation}
We must express the bound in terms of the positive-input weight $\nrm{\Pi_S \ket{\phi}}^2$. The two cases distinguish whether this weight controls the full squared norm or whether the component on label $0$ forces $\Delta_{\mathrm C} < 0$.

\begin{enumerate}
\item Suppose $\nrm{\Pi_S \ket{\phi}}^2 \geq \nrm{\phi}^2/16$. Then $\nrm{b}^2 \leq \nrm{\phi}^2 \leq 16 \nrm{\Pi_S \ket{\phi}}^2$. Combining the quantum bound with \eqref{eq:andor-classical-decrease} gives
\[
\Delta_{s, \nu} \leq \frac{16 - 8 \rho}{\sqrt n} \nrm{a} \nrm{b} - 2 \rho \nrm{a}^2 + \frac{2 \rho}{n} \nrm{b}^2.
\]
Maximizing the quadratic in $\nrm{a}$ and using $0 < \rho \leq 1$ gives
\[
\Delta_{s, \nu} \leq \frac{32}{\rho n} \nrm{b}^2 \leq \frac{512}{\rho n} \nrm{\Pi_S \ket{\phi}}^2.
\]

\item Suppose $\nrm{\Pi_S \ket{\phi}}^2 < \nrm{\phi}^2/16$. The support assumption puts more than $15/16$ of the squared norm on label $0$. Since $(I - \Pi_S) \ket{\phi}$ has only that label,
\[
\nrm{\Pi_{v_S}^{\perp}(I - \Pi_S) \ket{\phi}}^2 = \frac{r}{n + r} \nrm{(I - \Pi_S) \ket{\phi}}^2.
\]
Apply $\nrm{\ket{u} + \ket{v}}^2 \geq \frac 1 2 \nrm{u}^2 - \nrm{v}^2$ to the projections of the label-$0$ and positive-label components onto the orthogonal complement of $\ket{v_S}$. Using $r/(n + r) > 2/5$ gives
\[
\nrm{\Pi_{v_S}^{\perp} \ket{\phi}}^2 \geq \frac{r}{2(n + r)} \nrm{(I - \Pi_S) \ket{\phi}}^2 - \nrm{\Pi_S \ket{\phi}}^2 > \frac 1 8 \nrm{\phi}^2.
\]
The raw progress is therefore at least $\frac 1 2 \nrm{\phi}^2$. The classical outcome contributions sum to at most $\frac 1 4 \nrm{\phi}^2$, so $\Delta_{\mathrm C} \leq -\frac 1 4 \nrm{\phi}^2$. Using $\nrm{a} \leq \nrm{\Pi_S \ket{\phi}}$ and $\nrm{b} \leq \nrm{\phi}$ in the quantum bound gives
\[
\Delta_{s, \nu} \leq \frac{16}{\sqrt n} \nrm{\Pi_S \ket{\phi}} \nrm{\phi} - \frac{\rho}{4} \nrm{\phi}^2 \leq \frac{256}{\rho n} \nrm{\Pi_S \ket{\phi}}^2.
\]
The last step maximizes the quadratic in $\nrm{\phi}$.
\end{enumerate}
\end{proof}

We now sum over the query indices and phase bits. Since the raw form acts on $\mathcal Z$ and the query projectors have orthogonal ranges,
\[
\Psi_S^{\raw}(\ket{\phi_h})
= \sum_{i', j', \nu} \Psi_S^{\raw}\bigl((I_{\mathcal Z} \otimes P_{i', j', \nu}) \ket{\phi_h}\bigr).
\]
The single quantum outcome contributes at most $\Psi_S^{\raw}(\ket{\widehat\phi_{h\mathrm Q}})$, which decomposes in the same way. Classical outcomes at different indices or phase bits have different histories, so their capped contributions are added separately. At indices outside block $i$, all compared inputs agree. At a position outside $S$ in block $i$, every input with a nonzero component has bit $0$ at that position, by \eqref{eq:andor-support}. On each such component, the quantum query applies a common phase and the classical outcome leaves $S$ unchanged. The weighted raw increment is therefore zero. Thus
\begin{equation}
\sum_{e} \Psi_{he}(\ket{\widehat\phi_{he}})
\leq \Psi_h(\ket{\phi_h}) + \sum_{s \in S} \sum_{\nu = 0}^1 \Delta_{s, \nu}.
\label{eq:andor-aggregate}
\end{equation}

Orthogonality of the input labels in \eqref{eq:andor-comparison-vectors}, together with $\sum_{s \in S, \nu} P_{i, s, \nu} \preceq P_i$, gives
\[
\sum_{s \in S} \sum_{\nu = 0}^1 \nrm{\Pi_S \ket{\phi_{s, \nu}}}^2 = \frac{1}{2 n} \sum_{j \in S} \sum_{s \in S} \sum_{\nu = 0}^1 \nrm{P_{i, s, \nu} \ket{\psi_h^j}}^2 \leq \frac{1}{2 n} \sum_{j = 1}^n \nrm{P_i \ket{\psi_h^j}}^2.
\]
Apply \Cref{lem:andor-position} in \eqref{eq:andor-aggregate}, and use \eqref{eq:andor-saturated} for saturated vectors. Summing over $h: |h| = t$ and applying \eqref{eq:andor-processing} proves \eqref{eq:andor-query-increment}.

\subsection{Progress at termination}
\label{sec:andor-final}

\begin{proof}[Proof of \eqref{eq:andor-final-progress}]
Write $H^2(p, q) = 1 - \sum_o \sqrt{p(o)q(o)}$ for squared Hellinger distance \cite[Section~VI]{andor-topsoe}. For unit vectors $\ket{u}, \ket{v}$ and a common measurement with output distributions $p, q$, the comparison with $M$ and the measurement bound for fidelity \cite[Theorem~3.24]{andor-watrous} give
\begin{equation}
M(p, q) \leq 2 H^2(p, q) \leq 2(1 - |\langle u \mid v\rangle|) \leq \nrm{\ket{u} - \ket{v}}^2.
\label{eq:andor-measurement}
\end{equation}
For input $x_j$, combine the output vectors into $\ket{\psi_{\outlabel}^j} = \bigoplus_{|h| \leq T} \ket{\psi_{h, \outlabel}^j}$. By \eqref{eq:new} and termination this is a unit vector, and its output measurement has distribution $p_{x_j}$. Applying \eqref{eq:andor-measurement} and expanding the direct-sum norm gives
\begin{equation}
\frac{1}{n} \sum_{j = 1}^n M(p_{x_0}, p_{x_j})
\leq \sum_{|h| \leq T} \frac{1}{n} \sum_{j = 1}^n \nrm{\ket{\psi_{h, \outlabel}^0} - \ket{\psi_{h, \outlabel}^j}}^2.
\label{eq:andor-output-sum}
\end{equation}

To compare these squared differences with the progress, define, for $\ket{v} = \sum_{j = 0}^n \ket{j} \ket{v_j}$,
\[
D(\ket{v}) = \sum_{j = 1}^n \left\lVert\frac{1}{\sqrt{n}} \ket{v_0} - \ket{v_j}\right\rVert^2.
\]
For an output comparison vector,
\[
D(\ket{\phi_{h, \outlabel}}) = \frac{1}{2 n} \sum_{j = 1}^n \nrm{\ket{\psi_{h, \outlabel}^0} - \ket{\psi_{h, \outlabel}^j}}^2.
\]
Thus \eqref{eq:andor-output-sum} reduces the claim to $D(\ket{\phi_{h, \outlabel}}) \leq 8 \Psi_h(\ket{\phi_{h, \outlabel}})$ for every history $h$. The inequality $D(\ket{v}) \leq 2 \nrm{v}^2$ proves this when $\ket{\phi_{h, \outlabel}}$ is saturated.

For nonempty $S \subseteq [n]$ and a vector $\ket{v}$ supported on $\{0\} \cup S$, the terms in $D$ with $j \notin S$ contribute $(n - |S|) \nrm{v_0}^2/n \leq \lambda_S \nrm{v}^2$. The sum over $j \in S$ vanishes on the common component along $\ket{v_S}$, so its value is unchanged after applying $\Pi_{v_S}^{\perp}$. Its value is at most $D(\Pi_{v_S}^{\perp} \ket{v}) \leq 2 \nrm{\Pi_{v_S}^{\perp} \ket{v}}^2$. Hence
\[
D(\ket{v}) \leq \lambda_S \nrm{v}^2 + 2 \nrm{\Pi_{v_S}^{\perp} \ket{v}}^2.
\]
Apply this with $S = S_h$ and $\ket{v} = \ket{\phi_{h, \outlabel}}$ when the output vector is unsaturated. Then $\lambda_S < 1/4$, and the right-hand side is at most $\Psi_S^{\raw}(\ket{v}) = \Psi_h(\ket{v})$. This completes the output comparison.
\end{proof}

\section{Optimal hybrid unbiased phase and amplitude estimation}\label{sec:UnbiasedEst}

Phase and amplitude estimation protocols~\cite{kitaev1995QMeasAndAbelianStabilizer,brassard2002AmpAndEst} are fundamental building blocks of quantum algorithms. Despite much recent effort in making these primitives more resource friendly~\cite{suzuki2019AEWithoutPE,grinko2019IQAE, fukuzawa2022ModifiedIQAE,magano2022SimplifyingClassQuantInterpolationQSVT,ni2023LowDepthPE,huang2026EigenstatePrepEigenprobFiltering}, there are only a handful of variants that achieve the asymptotically optimal scaling even in the fully coherent regime~\cite{aaronson2019QApproxCountingSimplified,fukuzawa2022ModifiedIQAE,mande2023TightBoundsPE}. Similarly, we are only aware of a few works that claim to achieve the optimal scaling~(\Cref{cor:phaseEstLower,cor:ampEstLower}) in a depth-bounded (hybrid) setting. For the problem of phase estimation, algorithms achieving the optimal trade-off in some regimes of \Cref{cor:phaseEstLower} were introduced in~\cite{wang2018AcceleratedVQE,belliardo2020AchievingHeisenbergScaling}. For the amplitude estimation problem, \cite{giurgicaTiron2020LowDepthAE} proposed two algorithms. The ``QoPrime'' algorithm provably achieves the optimal trade-off for depths $D=\Theta(\epsilon^{-(1-1/m)})$ with $m\geq 2$ a fixed integer, albeit with a prefactor growing with $m$. The ``Power-law Amplitude Estimation'' algorithm claims the optimal trade-off for all depths; however, its analysis~\cite[Theorem~2.3]{giurgicaTiron2020LowDepthAE} relies on unproven regularity conditions~\cite[Appendix~A]{giurgicaTiron2020LowDepthAE}. Neither algorithm explicitly addresses the improved complexity achievable for small values of $p$ (as in \Cref{cor:ampEstLower}).

Here we build upon a conceptually cleaner family of protocols constructed via bounded variance unbiased phase and amplitude estimation procedures~\cite{apeldoorn2022QTomographyWStatePrepUnis,cornelissen2023SubLinQAlgPartFuncEst,vu2024LowDepthAEWithoutReallyTrying}. Beyond conceptual clarity, these have several further advantageous features: they provide unbiased estimators\footnote{In case of implementation with finite resources, their bias can be exponentially reduced in terms of the employed gate complexity, see \cite[Section~6.4]{apeldoorn2022QTomographyWStatePrepUnis}.} and achieve the optimal asymptotic hybrid scaling~(\Cref{cor:phaseEstLower,cor:ampEstLower}) in all parameter regimes while the protocols are non-adaptive, so all depth-$q$ blocks can be executed in parallel, and the classical post-processing is straightforward.

Our optimal hybrid phase and amplitude estimation algorithms are built on the (boosted) unbiased phase estimation protocol of~\cite[Algorithm~2]{apeldoorn2022QTomographyWStatePrepUnis} via a careful analysis of its first few moments, removing the $\exp(-m)$ term present in the analysis of~\cite{apeldoorn2022QTomographyWStatePrepUnis}.

\begin{algorithm}[!ht]
	\caption{Boosted Unbiased Phase Estimation}\label{alg:unbiased}
	\begin{algorithmic}[1] 
		\STATEx	{\bf Input:}\label{line:Input} $2m+1$ copies of $\ket{\psi(\phi)}=\frac{1}{\sqrt{M}}\sum_{k=0}^{M-1}e^{i \phi k}\ket{k}$ (for unknown $\phi$)
		\STATE {\bf For} $j=1$ to $2m+1$
		\STATE ~~~\label{line:randomXi} Sample a uniformly random number $u_j\in[0,1)$ and define $\xi:=\frac{2\pi u}{M}$
		\STATE ~~~\label{line:applyPhase} Apply multi-phase gate $\sum_{k=0}^{M-1}e^{-i\xi k}\ketbra{k}{k}$ to $\ket{\psi(\phi)}_j$
		\STATE ~~~ Perform inverse Fourier transform over $\mathbb{Z}_M$ and measure the state yielding outcome $\ell\in\{0,1,\ldots,M\!-\!1\}\!$
		\STATE ~~~\label{line:correctPhase} {\bf Record} $\varphi_j:=\frac{2\pi \ell}{M}+\xi=\frac{2\pi}{M}(\ell+u_j)$ 
		\STATE\label{line:shortest}Find the shortest interval $I\!=\![a,b]\!\subseteq\! [-2\pi,2\pi]$ such that $I\cup (I+2\pi)$ contains at least $m+1$ of the estimates~$\varphi_j$  
		\STATE\label{line:averageOutput}{\bf If} $a+b\geq 0$ {\bf then return} $\overline{\varphi}:=\frac{a+b}{2}$ {\bf else return} $\overline{\varphi}:=\frac{a+b}{2}+2\pi$					
	\end{algorithmic}
\end{algorithm}	
\noindent The estimates $\varphi_j$ in Line~\ref{line:correctPhase} are independent and identically distributed. As shown by \cite[Theorem 6.1]{apeldoorn2022QTomographyWStatePrepUnis} the probability density $f(\varphi_j)$ of the phase estimate $\varphi_j\in[0,2\pi)$ in terms of its deviation $|\phi-\varphi_j|\in[-\pi,\pi]$ is 
\begin{equation}\label{eq:probDens}
	f(\varphi_j):=\frac{M}{2\pi}\frac{\sinc^2(\frac{M}{2}|\phi-\varphi_j|)}{\sinc^2(\frac{1}{2}|\phi-\varphi_j|)},
\end{equation}
whose error distribution is symmetric and independent of $\phi$, hence the name ``unbiased''; also in the absence of boosting ($m=0$), it is shown in \cite[Section 6.2]{apeldoorn2022QTomographyWStatePrepUnis} that  $Z_0:=\frac{M}{M-1}\exp(i\varphi_1)$ is an unbiased estimator of  $\mathbb{E}[Z_0]=\exp(i\phi)$, and the calculations therein also reveal the covariance matrix of the phase deviation is
\begin{align}\label{eq:m0Cov}
	\operatorname{Cov}[Z_0/\exp(i\phi)]=
	\left(
	\begin{array}{cc}
		\frac{1}{M-1}& 0\\ 0 & \frac{M}{(M-1)^2}
	\end{array}
	\right),
\end{align}
where we think about $Z_0$ as a 2D vector whose coordinates are the real an imaginary parts respectively.
\begin{lemma}\label{lem:varBound}
	For $m\geq 4$ and $M\geq 30$ the boosted estimator $Z_m:=\frac{\exp(i\varphi)}{|\mathbb{E}[\exp(i\varphi)]|}$ is unbiased $\mathbb{E}[Z_m]=\exp(i\phi)$ and the covariance matrix of the phase deviation is
\end{lemma}
\begin{align}\label{eq:BoostedCov}
	\operatorname{Cov}[Z_m/\exp(i\phi)]\preceq\left(1-\sqrt{\frac{m-2}{m-3}}\left(\frac{8\pi}{M}\right)^{\!\!2}\right)^{\!\!-2}
	\cdot
	\left(
	\begin{array}{cc}
		\frac{m-2}{m-3}\left(\frac{8\pi}{M}\right)^{\!4}& 0\\ 0 & \frac{2m-1}{2m-2}\left(\frac{8\pi}{M}\right)^{\!2}
	\end{array}
	\right).
\end{align}
Note that the improved variances $\bigO(\frac{1}{M^4})$, $\bigO(\frac{1}{M^2})$ corresponding to the real and imaginary parts respectively are sufficiently concentrated to realize coherent metrologic advantage needed for achieving the optimal trade-off.
\begin{proof}
The off-diagonal entries of \eqref{eq:m0Cov} are zero due to the symmetry of the distribution. Since the boosting in Line~\ref{line:shortest} is also performed in a symmetric way, we have that for all $m\geq 0$ the covariance of the rotated version of the boosted estimator $Z_m$ has also vanishing off-diagonals.\footnote{$\mathbb{E}[|\exp(i\varphi)|]$ is a constant that does not depend on $\phi$ because the distribution of $\phi-\varphi_1$ does not depend on $\phi$ either, and the boosting is defined in a phase-shift-invariant manner.} Let $\theta:=\phi-\varphi$ and $\theta_j:=\phi-\varphi_j$; now we bound the variance of $\operatorname{Im}(\exp(i\theta))=\sin(\theta)$, i.e., the bottom right element of \eqref{eq:BoostedCov} for $m\geq2$.

We begin by computing a tail bound on $|\theta_j|$. For all $r\in[0,\pi]$ we have
\begin{align}
	\mathbb{P}[|\theta_j|\geq r]&=\int_r^\pi\frac{M}{\pi}\frac{\sinc^2(\frac{M}{2}\alpha)}{\sinc^2(\frac{1}{2}\alpha)}d\alpha\tag{by~\eqref{eq:probDens}}\\&
	\leq \int_r^\pi\frac{\pi M}{4}\frac{\sin^2(\frac{M}{2}\alpha)}{(\frac{M}{2}\alpha)^2}d\alpha\tag{since~$|\sinc(\alpha/2)|\geq\frac{2}{\pi}$}\\&
	\leq \int_r^\infty\frac{\pi}{M\alpha^2}d\alpha\tag{since~$\sin^2(\frac{M}{2}\alpha)\leq1$}\\&
	=\frac{\pi}{M r}\label{eq:thetaTail}
\end{align}

The other key observation is that for all $r\in[0,\pi]$ 
\begin{align}\label{eq:TailBound1}
	\text{if }\quad \left|\left\{j\in[2m+1]\colon |\theta_j|\leq\frac{r}{2}\right\}\right|\geq m+1\quad \text{ then }\quad |\theta|\leq r.
\end{align}
Indeed, then the interval $[-\frac{r}{2},\frac{r}{2}]$ contains at least $m+1$ of the $\theta_j$-s, thus the shortest such interval $I$ from Line~\ref{line:shortest} has length at most $r$. Also, by the pigeonhole principle $I$ must overlap with $[-\frac{r}{2},\frac{r}{2}]$, which then implies that the output in Line~\ref{line:averageOutput}, the average of the endpoints of $I$, satisfies $|\theta|\leq r$.
Thus, a tail bound on $\theta$ follows as
\begin{align}
	\mathbb{P}[|\theta|> r]&\leq \mathbb{P}\left[\left\{j\colon |\theta_j|>\frac{r}{2}\right\}\ge m+1\right]\nonumber\tag{by~\eqref{eq:TailBound1}}\\&
	\leq\sum_{k=m+1}^{2m+1}\binom{m+1}{2m+1}\left(\mathbb{P}\left[|\theta_1|>\frac{r}{2}\right]\right)^{\!m+1}\nonumber\\&
	\leq\sum_{k=m+1}^{2m+1}\binom{m+1}{2m+1}\frac{2\pi}{M r}\tag{by~\eqref{eq:thetaTail}}\\&
	\leq2^{2m}\left(\frac{2\pi}{M r}\right)^{\!\!m+1}
	=\frac{1}{4}\left(\frac{8\pi}{M r}\right)^{\!\!m+1}\label{eq:TailBound2}
\end{align}

Building on the above tail bound now we can bound the bottom right element of \eqref{eq:BoostedCov} for $m\geq 2$ as
\begin{align}
	\operatorname{Var}[\sin(\theta)]=
	\mathbb{E}[\sin^2(\theta)]&
	\leq \mathbb{E}[|\theta|^2]\tag{since $|\sin(\theta)|\le |\theta|$}\\&
	=\int_{0}^{\infty}2r\mathbb{P}[|\theta|>r]dr\tag*{$\left(\mathbb{E}\left[|X|^p\right] = \int_0^\infty p\,x^{p-1}\,\mathbb{P}(|X|>x)\,dx.\right)$}\\&
	\leq \int_{\frac{8\pi}{M}}^{\infty}2r\mathbb{P}[|\theta|>r]dr + \left(\frac{8\pi}{M}\right)^{\!\!2}\nonumber\\&
	\leq \int_{\frac{8\pi}{M}}^{\infty}\frac{r}{2}\left(\frac{8\pi}{Mr}\right)^{\!\!m+1}dr + \left(\frac{8\pi}{M}\right)^{\!\!2}\tag{by~\eqref{eq:TailBound2}}\\&
	=\frac{2m-1}{2m-2}\left(\frac{8\pi}{M}\right)^{\!\!2}.\label{eq:ImVar}
\end{align}
A similar argument, valid for $m\ge 4$ bounds the top left element of \eqref{eq:BoostedCov} as follows
\begin{align}
	\operatorname{Var}[\cos(\theta)]
	&\le \mathbb{E}[(1-\cos(\theta))^2]\tag{$\operatorname{Var}[X]\leq \mathbb{E}[X^2]$}\\&
	\leq \mathbb{E}[|\theta|^4/4]\tag{since $1-\cos(\theta)\le \theta^2/2$}\\&
	=\int_{0}^{\infty}4r^3\mathbb{P}[|\theta|>r]dr\tag*{$\left(\mathbb{E}\left[|X|^p\right] = \int_0^\infty p\,x^{p-1}\,\mathbb{P}(|X|>x)\,dx.\right)$}\\&
	\leq \int_{\frac{8\pi}{M}}^{\infty}4r^3\mathbb{P}[|\theta|>r]dr + \left(\frac{8\pi}{M}\right)^{\!\!4}\nonumber\\&
	\leq \int_{\frac{8\pi}{M}}^{\pi}r^3\left(\frac{8\pi}{Mr}\right)^{\!\!m+1}dr + \left(\frac{8\pi}{M}\right)^{\!\!4}\tag{by~\eqref{eq:TailBound2}}\\&
	=\frac{m-2}{m-3}\left(\frac{8\pi}{M}\right)^{\!\!4}.\label{eq:RemVar}
\end{align}
By Chebyshev's inequality we get 
\begin{align}
	1-|\mathbb{E}[\exp(i\varphi)]|
	=\mathbb{E}[1-\cos(\theta)]
	\leq \sqrt{\mathbb{E}[(1-\cos(\theta))^2]}
	\leq \sqrt{\frac{m-2}{m-3}}\left(\frac{8\pi}{M}\right)^{\!\!2}.\label{eq:LenBound}
\end{align}
	We obtain \Cref{eq:BoostedCov} by combining \eqref{eq:ImVar},\eqref{eq:RemVar}, and \eqref{eq:LenBound}.
\end{proof}
As shown in \cite[Section 6.2]{apeldoorn2022QTomographyWStatePrepUnis}, this suffices to achieve the optimal scaling of \Cref{cor:phaseEstLower} and \Cref{cor:ampEstLower} with using the above unbiased estimators --- in the case of amplitude estimation by taking the real part of the phase estimate coming from the Grover iterate. For completeness, we describe the details of the algorithms in as follows.
\begin{corollary}\label{cor:unbiasedPE}
	Given an eigenstate $U \ket{\psi} = e^{i \phi} \ket{\psi}$ we can prepare the state $\ket{\psi(\phi)}\ket{\psi}=\frac{1}{\sqrt{q+1}}\sum_{k=0}^{q}e^{i \phi k}\ket{k}\ket{\psi}$ with $q$ application of the controlled-$U$ operation. 
	By repeating \Cref{alg:unbiased} (with $m=0$) $\ell\geq 9$ times we obtain an unbiased estimator of $\exp(i\phi)$ of variance $\bigO(\frac{1}{q^2\ell})$, which then solves the problem of \Cref{cor:phaseEstLower} for $\ell=\Theta(\frac{1}{q^2\epsilon^2})$.
\end{corollary}
\begin{proof}
	For a single run of \Cref{alg:unbiased} the variance For $q<30$ the stated variance bound  (with $m=0$) follows from \Cref{eq:m0Cov}. Otherwise, by grouping the subsequent outcomes to $9$-tuples and post-processing them as in Line~\ref{line:shortest}, (i.e., effectively running \Cref{alg:unbiased} with $m=4$), each tuple produces an unbiased estimator of variance $\bigO(\frac{1}{q^2})$ due to \Cref{lem:varBound}. Since the subsequent estimates are independent, by taking their average, we get an $\bigO(\frac{1}{q^2\ell})$ variance unbiased estimator of $\exp(i\phi)$. By Chebyshev's inequality, for $\ell=\Theta(\frac{1}{q^2\epsilon^2})$ this results in an unbiased estimator of $\exp(i\phi)$ that produces an estimate within precision $\epsilon$ with probability at least $\frac23$.
\end{proof}
Note that in the above proof we did not attempt to optimize constants. We believe, that by optimizing parameters, the constants can be made competitive with alternative approaches in the literature. Indeed, for example, instead of $m=4$ we could take $m=2$, as then \Cref{eq:ImVar} already holds. Moreover, in practice, it is probably a lot more efficient to consider all runs, and take the smallest arc $I$ that contains at least half of the estimates, and take the average those, slightly modifying the boosting scheme of \Cref{alg:unbiased}.

\begin{corollary}\label{cor:unbiasedAE}
	Let $U \ket{\bar{0}} = \sqrt{p}\ket{1}\ket{\psi_1}+\sqrt{1-p}\ket{0}\ket{\psi_0}$ as in \Cref{cor:ampEstLower}. If $q=\bigO(1)$ we get an unbiased estimator of $p$ with variance $p(1-p)=\bigO(\frac{p(1-p)}{q^2})$ by applying $U$ to $\ket{\bar{0}}$, measuring, and outputting the value of the first qubit. 
	
	In the $1 \ll q \leq (p(1-p))^{-\frac12}$ regime we can get an unbiased estimator of $T_q(p)$ ($T_q$ being the q-th Chebyshev polynomial) by using amplitude amplification $\bigO(q)$-times, measuring, and outputting the value of the first qubit. Taking $\ell=\Theta(\frac{p(1-p)}{q^2\epsilon^2})$ samples yields a variance $\bigO(\epsilon^2/q^2)$ unbiased estimator of $T_q(p)$, which produces a (potentially biased) estimator of $p\approx T_q(p)/q$ with precision $\approx\epsilon$ with probability at least $2/3$.
	
	If $q\geq \max(9\cdot 30,(p(1-p))^{-\frac12})$ we can prepare the state $\frac{1}{\sqrt{2}}\sum_{s\in\pm }\ket{\psi_s}(\ket{\psi(s\phi)})^{\otimes 9}$, where $\ket{\psi}=\frac{1}{\sqrt{2}}(\ket{\psi_{+}}+\ket{\psi_{-}})$, $\phi=2\arcsin(\sqrt{p})$, and $\ket{\psi(\pm\phi)}$ is as in the input line of \Cref{alg:unbiased} for $M=1+\lfloor\frac{ q }{9}\rfloor$ by applying a controlled version of the unitary $U^\dagger(Z\otimes I)U(I-2\ketbra{\bar{0}}{\bar{0}})$. Running \Cref{alg:unbiased} with $m=4$ and taking the real part of the output $\varphi$ is an unbiased estimator of $\cos(\pm \phi)=1-2p$ (and hence $p$) with variance $\bigO(\frac{p(1-p)}{q^2}+\frac{1}{q^4})$.
	Repeating either of the above procedure $\ell$ times we obtain an unbiased estimator of $p$ with variance $\bigO(\frac{p(1-p)}{q^2\ell}+\frac{1}{q^4\ell})$, which then solves the problem of \Cref{cor:ampEstLower} for $\ell=\Theta(\frac{p(1-p)}{q^2\epsilon^2})$.
 \end{corollary}
\begin{proof}
	For $q< 9 \cdot 30$ the variance is just the variance of he associated Bernoulli random variable corresponding to probability $p$. 
	
	In the $1\ll q \leq (p(1-p))^{-\frac12}$ regime our analysis follows from the simple observation that $T_q(p)\approx q p$.
	
	If $q\geq \max(9\cdot 30,(p(1-p))^{-\frac12})$ the decomposition and the eigenvalues come from the analysis of the original amplitude estimation protocol~\cite{brassard2002AmpAndEst}.
	
	The stated variance bound for a single run of \Cref{alg:unbiased} (with $m=4$) follows from \Cref{lem:varBound}, by ``rotating'' the covariance matrix \Cref{eq:BoostedCov} with angle $\pm\phi$. Since $\sin^2(\pm \phi)=4p(1-p)$ we get the stated $\bigO(\frac{p(1-p)}{q^2}+\frac{1}{q^4})$ variance bound.
	
	 Since the subsequent estimates are independent, by taking their average, we get an $\bigO(\frac{p(1-p)}{q^2\ell}+\frac{1}{q^4\ell})$ variance unbiased estimator of $p$. By Chebyshev's inequality, for $\ell=\Theta(\frac{p(1-p)}{q^2\epsilon^2})$ this results in an unbiased estimator of $p$ that produces an estimate within precision $\epsilon$ with probability at least~$\frac23$.
\end{proof}

\section{Conclusion}

In this work, we developed methods for proving lower bounds on hybrid quantum-classical query algorithms. Our methods are based on a $\chi^2$-type progress measure on the probability distributions of classical 
outcomes of the algorithm. For binary classification of unitaries and search, we developed an inequality that relates the change in progress measure to the distance between final states of a quantum subroutine on the respective inputs. This allowed to handle the problems in a uniform way and has a potential for generalization. For AND-OR, we used a more intricate argument that combines a number of ideas.

As the study of quantum lower bounds in the hybrid model is just starting,
there is a number of open problems in the area. We list some of them below.
\begin{enumerate}
\item
{\bf More direct proof of Lemma \ref{lem:one-block}.}
In its essence, Lemma \ref{lem:one-block} is a very natural statement, as it upper bounds the distinguishability of probability distributions on transcripts via probabilities of
querying the variables that are different in the respective two inputs. This raises a question: can one prove Lemma \ref{lem:one-block} similarly to statements in earlier sections,
by upper bounding the increase in distinguishability after one block by queries via probability of querying the relevant variables?
\item
{\bf AND-OR trees of larger depth.}
Our methods extend to tight lower bounds for AND-OR trees of a larger constant depth, with big-$\bigO$ constants that increase as the tree depth grows. Can one obtain a tight lower bound for AND-OR trees of arbitrary depth?
\item
{\bf Approximate counting.} A natural target for extending our lower bound methods is approximate counting (for example, distinguishing whether the fraction of variables $x_i$ that are equal to 1 is
$\frac{1}{2}$ or $\frac{1}{2}+\epsilon$). Similarly to AND-OR, the optimal quantum lower bound is proven by a positive weight adversary method with a simple distribution over inputs. This indicates that our methods may be extendable to this case.
\item
{\bf Element distinctness.} Another well known problem is element distinctness. Several natural strategies of breaking down the quantum algorithm into smaller algorithms lead to the complexity of
$\bigO(\max(N^{2/3}, \frac{N^2}{q^2})$ queries which provides advantage over the classical algorithm only if $q>\sqrt{N}$. Is this optimal? There are two possible difficulties for proving the lower bound.
First, the quantum lower bound is more complex, as positive-weight adversary no longer gives an optimal result. Second, we feel that the complexity of the problem may be different in the hybrid and the
noisy query models and, therefore, it might not be possible to use the noisy query model for an optimal lower bound.
\end{enumerate}

\section*{Acknowledgments and AI disclosure}
The authors developed the key ideas and wrote the manuscript. LLMs were used only to aid in completing some details of standard technical proofs, finding relevant references, and providing feedback on the writing.
AA, AG, and MK are supported by the QuantERA project HQCC. AG is supported by the European Research Council (ERC) via the Starting grant QuanThermal (101222179).
AG is grateful for insightful discussions about bounded-depth phase and amplitude estimation with Jona Erle, Bálint Koczor, and Po-Wei Huang.

\bibliographystyle{alphaUrlePrint.bst}
\bibliography{qc_gily,ref}
\end{document}